\documentclass{ifacconf}

\usepackage{graphicx}      
\usepackage{natbib}        

\let\theoremstyle\relax
\usepackage{amsthm}
\usepackage{amssymb}
\usepackage{mathrsfs}
\usepackage{mathtools}
\usepackage[normalem]{ulem}
\usepackage{bm}
\usepackage{diagbox}
\usepackage{booktabs}
\usepackage{adjustbox}
\usepackage{multirow, multicol}
\usepackage{makecell}
\usepackage{comment}
\usepackage{algorithm}
\usepackage{algpseudocode}
\usepackage{color, colortbl}
\usepackage{subcaption}
\algrenewcommand\textproc{}
\usepackage{threeparttable}
\usepackage{enumitem}
\newcolumntype{M}[1]{>{\centering\arraybackslash}m{#1}}

\theoremstyle{plain}
\newtheorem{lemma}{\textbf{Lemma}}

\newtheorem{theorem}{\textbf{Theorem}}
\newtheorem{corollary}{\textbf{Corollary}}
\newtheorem{assumption}{Assumption}
\newtheorem{problem}{\textbf{Problem}}

\theoremstyle{definition}
\newtheorem{definition}{\textbf{Definition}}
\newtheorem{remark}{Remark}
\newtheorem{example}{\textbf{Example}}
\useunder{\uline}{\ul}{}

\newcommand{\Lc}{\mathcal{L}}

\newcommand{\Ac}{\mathcal{A}}

\DeclareMathOperator*{\Imag}{Im}
\DeclareMathOperator*{\Real}{Re}

\DeclarePairedDelimiter\ceil{\lceil}{\rceil}

\includecomment{noncounted} 
\begin{document}
\begin{frontmatter}

    \title{Data-driven Reachability Verification with Probabilistic
    Guarantees under Koopman Spectral Uncertainty\thanksref{footnoteinfo}}

    \thanks[footnoteinfo]{
        We acknowledge the financial support of the Finnish Ministry
        of Education and Culture through the Intelligent Work
        Machines Doctoral Education Pilot Program (IWM VN/3137/2024-OKM-4).
    }

    \author[First]{Jianqiang Ding}
    \author[First]{Shankar A. Deka}

    \address[First]{Department of Electrical Engineering, Aalto
    University, Finland. (e-mail: {\{jianqiang.ding, shankar.deka\}}@aalto.fi).}


\begin{abstract}
    Providing rigorous reachability guarantees for unknown complex
    systems is a crucial and challenging task. In this paper, we
    present a novel data-driven framework that addresses this
    challenge by leveraging Koopman operator theory.
    Instead of operating in the state space, the proposed method
    encodes model uncertainty from finite data directly into Koopman
    spectral representation with quantifiable error bounds.
    Leveraging this spectral information, we systematically determine
    time intervals within which trajectories from the initial set are
    guaranteed, with a prescribed probability, to reach the target
    set.
    We finally demonstrate the efficacy of our framework in
    numerical examples.
\end{abstract}


\end{frontmatter}


\section{Introduction}


With the increasing integration of complex autonomous systems, such
as autonomous vehicles, robotics, and smart grids, into
safety-critical scenarios, guaranteeing their operational safety and
reliability has become a paramount challenge in modern control engineering.
Fundamentally, verifying whether a dynamical system satisfies these operational
specifications can be framed as a reachability problem. For instance, collision
avoidance equates to ensuring that state trajectories never enter an unsafe set,
while task completion corresponds to guaranteeing that the system eventually reaches
a designated target set.

Historically,
reachability verification
problem
has been
extensively studied under the assumption that
a precise dynamical model is available.
Prominent approaches
include set-propagation methods using zonotopes
\citep{alanwar2023data},
backward reachability
analysis
via
Hamilton-Jacobi-Isaacs (HJI) equations \citep{herbert2019reachability},
and verification via the
construction of barrier certificates through Sum-of-Squares (SOS)
optimization techniques
\citep{prajna2007convex}.
However, this reliance on precise analytical models presents
a severe practical limitation. Deriving models from first principles for complex
real-world systems is often either impossible due to unknown dynamics, or results
in models too computationally intractable for formal analysis.
This fundamental gap between model-based theory and the realities of
complex engineering systems has spurred a paradigm shift to
data-driven approaches, which aim to infer system properties directly
from observed data, bypassing the identification of an explicit model.

Recent advancements in data-driven reachability analysis have
established a diverse landscape of methodologies.
For instance, scenario optimization approaches like
\citep{devonport2020estimating}
formulate the
determination of reachable set as a chance-constrained optimization
problem, which yields explicit probabilistic guarantees but often at
the cost of high sample complexity.
Surrogate-model approaches construct reachable set estimates
with probabilistic guarantees, using techniques such as Gaussian
processes \citep{devonport2020data} or conformal prediction
\citep{tebjou2023data}.
In contrast, methods
seeking
deterministic guarantees, such as
those using Zonotopes \citep{alanwar2021data}, focus on
constructing over-approximation of the true reachable set but are
typically conservative for verification. While emerging techniques
using Neural PDE solvers like DeepReach
\citep{bansal2020deepreach} promise high scalability, their assurances
currently remain largely empirical.
Fundamentally
,
these approaches
all embed
data-driven uncertainty into the geometric representation
of the reachable set, thereby inheriting the limitations of set-based
computations.
Even when inferring the dynamical model from finite noise-free data,
the intrinsic uncertainty presents
a critical problem in the field that remains largely open:
\textit{how can we, under model uncertainty inherent in finite data,
    rigorously verify the reachability of the underlying dynamical system
    in general settings, such as between non-convex sets or over long,
potentially infinite, time horizons?}

To address this challenge, this paper introduces a novel data-driven reachability
verification framework based on the Koopman spectrum. Diverging from conventional
methods, our paradigm \citep{ding2024time} encodes
the underlying dynamic uncertainty into Koopman spectral representations
equipped with rigorous error bounds.
Levering these spectral information,
we verify the reachability specification by computing the probability
of the existence of a non-empty time interval, during which
trajectories starting from an initial set can reach a target set.


\section{Preliminaries}

\textit{Notations:}
$\mathbb{R}$, $\mathbb{C}$, and $\mathbb{R}^n$ denote real
numbers, complex numbers, and $n$-dimensional real space, respectively.
$\mathcal{C}^k(X)$ denotes the space of $k$-times continuously
differentiable functions on $X$.
$\Real(\cdot), \Imag(\cdot),
|\cdot|, \angle\cdot$ denote the real part, imaginary part,
absolute value, and phase angle.
$\|\cdot\|$ is the Euclidean norm on $\mathbb{R}^n$, $\mathcal{B}_r(z)
\doteq \{x : \|x - z\| < r\}$ the open ball of radius $r$ at $z$, $\mathrm{vol}
(S)$ the Lebesgue measure of $S \subseteq \mathbb{R}^n$, and
$\omega_n \doteq \pi^{n/2}/\Gamma(n/2+1)$ the volume of
$\mathcal{B}_1(0)$.

In this paper, we consider continuous-time dynamical systems of the form
\begin{equation}
    \label{definition:dynamic system}
    \frac{d}{dt}x(t) = f(x(t))
\end{equation}
where $f: X \rightarrow \mathbb{R}^n$ is an unknown continuously
differentiable map and $x(t) \in \mathbb{R}^n$ denotes the state in
the compact set $X \in \mathbb{R}^n$ at time $t \geq 0$.
The flow map $s_t: X \rightarrow X$ of system
\eqref{definition:dynamic system} is given by $s_t(x)=x+\int_0^t
f(x(\tau)) d\tau$, for all $t \geq 0$ and $x \in X$.
Throughout this paper, the explicit form of dynamical system
\eqref{definition:dynamic system} is considered to be unknown.
\begin{assumption} \label{assumption:setup}
    We make the following assumptions,
    \begin{enumerate}[label=(\roman*), leftmargin=*, nosep]
        \item A dataset $\mathcal{D}=\{(\bm{x}_k,
            \bm{y}_k)\}_{k=1}^N$ is available with $\bm{y}_k =
            s_{\Delta t}(\bm{x}_k)$, $\Delta t > 0$, and $\bm{x}_k
            \stackrel{\textnormal{i.i.d.}}{\sim} \mu$ for some
            measure $\mu$ on the compact set $X$ that is absolutely
            continuous w.r.t.\ the Lebesgue measure.
        \item For each approximated principal eigenfunction
            $\tilde{\psi}$, $\log|\tilde{\psi}|$ and
            $\angle\tilde{\psi}$ are Lipschitz on $X$ with known
            upper bounds $\bar L_{\mathcal{L}}$ and $\bar
            L_{\mathcal{A}}$ on their Lipschitz constants.
    \end{enumerate}
\end{assumption}
Under Assumption \ref{assumption:setup},
we are interested in the problem formulated as follows.\\
\begin{problem}
    (Data-driven Reachability Verification)
    Consider a dynamical system \eqref{definition:dynamic system}
    known only through the dataset $\mathcal{D}$. Given an initial
    set $X_0 \subset X$ and a target set $X_F \subset X$,
    reachability verification aims to determine whether there exists
    some time $t \geq 0$ and $x_0 \in X_0$ such that $s_t(x_0) \in X_F$.
\end{problem}

In the following, we review the fundamentals of Koopman principal eigenpairs,
demonstrate how the derived reach-time bounds enable reachability
verification, and investigate the eigenfunction approximation error bounds.

\subsection{Koopman operator and principal eigenpairs}

Let $\mathcal{F}$ be a Banach space of scalar-valued functions
$\phi(\bm{x}): X \rightarrow \mathbb{C}$. The Koopman operator
$\mathbb{U}_t: \mathcal{F} \rightarrow \mathcal{F}$ corresponding to
the system \eqref{definition:dynamic system} is defined as
$[\mathbb{U}_t \phi](\bm{x}) = \phi(s_t(\bm{x}))$, where
$\phi(\bm{x}) \in \mathcal{F}$ denotes an observable function.
\begin{definition} (Koopman Eigenvalues and Eigenfunctions)
    An observable function $\phi(x) \in \mathcal{F}$ is said to be an
    eigenfunction of the Koopman operator corresponding to the
    eigenvalue $\lambda$ if
    \begin{equation}
        \label{eq: eigens}
        [\mathbb{U}_t \phi](x) = e^{\lambda t}\phi(x), \ t \geq 0.
    \end{equation}

    With the Koopman generator, $\mathcal{K}_f$, equation (\ref{eq:
    eigens}) can be written as
    \begin{equation}
        \label{eq: koopman eigens}
        \mathcal{K}_f \phi = \frac{\partial \phi}{\partial x} f(x) =
        \lambda \phi(x)
    \end{equation}

    In addition to being defined for all $t \in [0,\infty)$ and $x
    \in X$, the Koopman spectrum can be generalized to finite time
    and the subset of state space as open and subdomain eigenfunctions.
    Furthermore, we define the set of principal eigenpairs as the
    minimal generator $G$ of the set given by
    \begin{equation*}
        E =
        \left\{\left(\sum_{i=1}^{m}n_i\lambda_i,\prod_{i=1}^{m}\phi_i^{n_i}\right)
            \bigg\lvert \;(\lambda_i,\phi_i) \in G,\; n_i \in \mathbb{N}
        \right\}.
    \end{equation*}
    where $E$ is the semigroup of eigenpairs $(\lambda, \phi)$. \\
\end{definition}
\begin{remark}
    The concept of principal eigenfunctions introduced in
    \citep{mohr2016koopman, kvalheim2021existence}
    yields a countably infinite set of
    eigenfunctions. To form a richer basis, this notion is extended
    to `primary eigenfunctions' in
    \citep{bollt2021geometric}, by allowing real-valued exponents. In
    this paper, we adopt this more general primary eigenfunction
    definition for reachability verification.
\end{remark}


\subsection{Time-to-reach bounds using the Koopman spectrum}
Given any function $g: X \rightarrow \mathbb{C}$ and sets $V,W
\subset X$, we define the following notations for convenience,
\begin{gather}
    \overline{g}_{V} \doteq \sup_{x\in V}|g(x)|,\quad
    \underline{g}_{V} \doteq \inf_{x\in V}|g(x)|
    \notag \\
    \Lc^{g}_{W,V} \doteq \log\!\left(
    \frac{\overline{g}_{V}}{\underline{g}_{W}} \right), \quad
    \Ac^{g}_{W,V} \doteq \overline{\angle g}_{V} -
    \underline{\angle g}_{W}. \notag
\end{gather}
We recall the following results from \citep{ding2024time} for
reachability verification with reach-time bounds.
\begin{theorem}
    \label{th:relaxation} Let $V, W \subseteq X$ be non-empty
    compact sets, and let $I^{mag}_{V,W}(\lambda,\psi)$ be the
    time-to-reach bounds from $V$ to $W$ with magnitudes of
    non-trivial eigenpair $(\lambda,\psi) \in E$, where
    $\psi=\prod_{i=1}^{n}\psi_{i}^{\alpha_i}$,
    $\lambda = \sum_{i=1}^{n}\alpha_{i}\lambda_{i}$, $\alpha_i \ge
    0$, and $(\lambda_i,\psi_i) \in G$ are principal eigenpairs.
    Then $I^{mag}_{V,W}(\lambda,\psi) \subseteq \hat{I}^{mag}_{V,W}
    (\lambda,\psi)$ with
    \vspace{-0.5em}
    \begin{enumerate}
        \item[(a)] For $\Real{(\lambda)}>0$, $\hat{I}^{mag}_{V,W}(
            \lambda,\psi) \doteq$
            \[
                \hspace{0.8em}\left[
                    \frac{\sum_{i=1}^{n}\alpha_{i}\Lc^{\psi_i}_{W,V}}{
                    -\sum_{i=1}^{n}\alpha_{i}\Real{(\lambda_{i})}},\,
                    \frac{\sum_{i=1}^{n}\alpha_{i}\Lc^{\psi_i}_{V,W}}{
                    \sum_{i=1}^{n}\alpha_{i}\Real{(\lambda_{i})}}
                \right]
            \]

        \item[(b)] For $\Real{(\lambda)}<0$, $\hat{I}^{mag}_{V,W}(
            \lambda,\psi) \doteq$
            \[
                \hspace{0.8em}\left[
                    \frac{\sum_{i=1}^{n}\alpha_{i}\Lc^{\psi_i}_{V,W}}{
                    \sum_{i=1}^{n}\alpha_{i}\Real{(\lambda_{i})}},\,
                    \frac{\sum_{i=1}^{n}\alpha_{i}\Lc^{\psi_i}_{W,V}}{
                    -\sum_{i=1}^{n}\alpha_{i}\Real{(\lambda_{i})}}
                \right].
            \]
    \end{enumerate}
\end{theorem}

\begin{theorem}
    \label{th:complex_relaxation}
    Let $V, W \subseteq X$ be non-empty compact sets, and let
    $I^{phase}_{V,W}(\lambda,\psi)$ be the collection of
    time-to-reach bounds from $V$ to $W$ with phases of non-trivial complex
    eigenpair $(\lambda,\psi) \in E$,
    parameterized by $\psi = \prod_{i=1}^{n}\psi_{i}^{\alpha_i}$
    and $\lambda = \sum_{i=1}^{n}\alpha_{i}\lambda_{i}$ with
    $\alpha_{i}\ge 0$. Then $I^{phase}_{V,W}(\lambda,\psi)
    \subseteq \hat{I}^{phase}_{V,W}(\lambda,\psi)$ with
    \vspace{-0.5em}
    \begin{enumerate}
        \item [(a)] For $\Imag{(\lambda)}>0$ and some $m \in
            \mathbb{Z}$, $\hat{I}^{phase}_{V,W}(\lambda,\psi)
            \doteq$
            \begin{gather*}
                \Bigg[
                    \frac{\sum_{i=1}^{n}\alpha_{i}\Ac^{\psi_i}_{V,W} +
                    2m\pi}{\sum_{i=1}^{n}\alpha_{i}\Imag{(\lambda_i)}}
                    ,
                    \frac{-\sum_{i=1}^{n}\alpha_{i}\Ac^{\psi_i}_{W,V} +
                    2m\pi}{\sum_{i=1}^{n}\alpha_{i}\Imag{(\lambda_i)}}
                \Bigg]
            \end{gather*}

        \item[(b)] For $\Imag{(\lambda)}<0$ and some $m \in
            \mathbb{Z}$, $\hat{I}^{phase}_{V,W}(\lambda,\psi)
            \doteq$
            \begin{gather*}
                \Bigg[
                    \frac{\sum_{i=1}^{n}\alpha_{i}\Ac^{\psi_i}_{W,V} +
                    2m\pi}{-\sum_{i=1}^{n}\alpha_{i}\Imag{(\lambda_i)}}
                    ,
                    \frac{\sum_{i=1}^{n}\alpha_{i}\Ac^{\psi_i}_{V,W} +
                    2m\pi}{\sum_{i=1}^{n}\alpha_{i}\Imag{(\lambda_i)}}
                \Bigg].
            \end{gather*}
    \end{enumerate}
\end{theorem}

Let $\hat{I}_{V,W}(\lambda,\psi) \doteq \hat{I}^{mag}_{V,W}(
\lambda,\psi) \cap \hat{I}^{phase}_{V,W}(\lambda,\psi)$ denote the
computed reach-time bounds from $V$ to $W$ with eigenpair $(\lambda,\psi)$.
A necessary condition for reachability verification is then given by
the following corollary.
\begin{corollary}\label{cor:main}
    For the dynamical system \eqref{definition:dynamic system}, a
    necessary condition for a target set $X_{F} \subset X$ to be
    reachable from an initial set $X_{0} \subset X$ is that the
    intersection of all intervals $\hat{I}_{X_0,X_F}(\lambda,\psi)$
    is non-empty. In other words,
    \begin{gather*}
        \left\{s_t(x_0) \,|\, x_0 \in X_0 \right\} \bigcap X_F \neq
        \emptyset \text{ for some } t>0 \\
        \implies \bigcap_{(\lambda,\psi)\in E}
        \hat{I}_{X_0,X_F}(\lambda,\psi) \neq \emptyset.
    \end{gather*}
\end{corollary}

\begin{remark} \label{rem:vanish_decomp}
    Theorems \ref{th:relaxation} and \ref{th:complex_relaxation}
    involve $\log|\psi_i|$ and thus require $|\psi_i|$ to be bounded
    away from zero on $V \cup W$, vanishing $\psi_i$ degenerates
    $\hat{I}_{V,W}(\lambda_i, \psi_i)$ to $\mathbb{R}$ and is
    excluded.
    The proposed framework estimates the exact reach-time
    interval $I^*_{X_0, X_F}$ by intersecting the over-approximations
    $\hat{I}_{X_0, X_F}(\lambda, \psi)$ across all $(\lambda, \psi) \in E$, which
    tightens but does not eliminate the conservatism.
    Since the global error admits the exact decomposition
    \begin{equation*}
        \hat{I}_{X_0, X_F} \setminus I^*_{X_0, X_F} =
        \bigcup_{\bm{\alpha}}\bigl( I^e_{\bm{\alpha}} \cap
            \bigcap_{\bm{\beta} \neq \bm{\alpha}}
        \hat{I}_{\bm{\beta}} \bigr),
    \end{equation*}
    where $I^e_{\bm{\alpha}}$ denotes the error contributed by the
    composite eigenfunction with
    $\bm{\alpha}$,
    each $\hat{I}_{\bm{\alpha}}$ contributes to the global error whenever its
    over-approximation is not filtered out by the intersection with
    the other bounds. An inaccurate principal eigenpair therefore
    propagates to every composite in which it appears and inflates
    the estimate.
\end{remark}


\subsection{Error bounds on approximating Koopman spectral
properties}\label{subsec:Koopman_spectral_errors}
The data-driven approximation of Koopman spectral properties is
fundamentally rooted in the approximation of the Koopman operator itself.
Spurred by the evolution from Dynamic Mode Decomposition (DMD) \citep{schmid2010dynamic}
to Extended DMD (EDMD) \citep{williams2015data} and kernel-based variants
\citep{williams2014kernel,klus2020kernel}, this field has matured significantly,
shifting from exploratory algorithms to comprehensive frameworks grounded
in rigorous error analysis.
Seminal work by \citep{korda2018linear} proved the convergence of EDMD
to the true operator in the finite-data limit.
Probabilistic bounds with convergence rates for i.i.d. and ergodic
sampling, respectively, were established in \citep{philipp2023error}.
More recently,
$L^\infty$ error bounds \citep{kohne2025error} for
kEDMD
through interpolation theory.
However, an accurate operator approximation does not guarantee
accurate spectral properties, a challenge primarily attributed to
spectral pollution.
A critical development to address this challenge is Residual DMD
(ResDMD) \citep{colbrook2024rigorous}.
This method provides a framework
for computing the spectrum free of spectral pollution
and offers a computable metric to evaluate the residual of each eigenpair.
A key result, Theorem 4.1 in
\citep{colbrook2024rigorous}, establishes that the computed
eigenvalues converge to the true spectrum as the amount of data and
the dictionary size increase.
The theoretical origins of spectral pollution were further elucidated
in \citep{kostic2023sharp} by the concept of metric distortion, which
reveals that spectral error is amplified by a geometric factor
related to the chosen function space.
For systems with local analytical structure, high-precision methods
such as Analytical EDMD in \citep{mauroy2024analytic} and JetDMD in
\citep{ishikawa2024koopman} have been developed, leveraging Taylor or
jet-based expansions to achieve strong convergence guarantees for the
eigenfunctions.
These theoretical advancements establish that, under suitable
assumptions
regarding
data and basis functions, the error of a
data-driven Koopman eigenpair approximation can be rigorously
bounded. To formalize our results under spectral uncertainty
as illustrated in Assumption \ref{assumption:setup}(ii), we
consider the corresponding
eigenfunction approximation error in the
space $\mathcal{F}$ to be bounded by a constant $\delta_\psi \geq 0$,
\begin{align} \label{eq:eigenfunction_error_bound}
    \| \varepsilon(x) \psi(x) -  \psi(x) \|_{\mathcal{F}} \leq \delta_\psi
\end{align}
where $\varepsilon(x)$ represents the local relative error of
$\psi(x)$ to its approximation $\tilde{\psi}(x)$
after optimal scaling alignment
.

\section{Data-driven Reachability Verification under Spectral Uncertainty}

In this section,
we present the main technical contribution of this paper by deriving
rigorous formal guarantees for reachability verification with
uncertain principal Koopman eigenpairs.
By quantifying the combined impact of uncertainties from sampling and eigenpairs,
we establish probabilistic guarantees on the reach-time bounds.
Given an uncertain Koopman eigenpair
$(\tilde{\lambda},\tilde{\psi})$ satisfying
Assumption \ref{assumption:setup}(ii)
and bounded as in
\eqref{eq:eigenfunction_error_bound}, the quantities
$\mathcal{L}^{\tilde{\psi}}$ and $\mathcal{A}^{\tilde{\psi}}$ can be
bounded as follows,
\color{red}
\color{black}

\begin{lemma} \label{lemma:impact of inaccurate eigenpairs}
    Suppose the inaccurate eigenfunction $\tilde{\psi}(x) =
    \varepsilon(x) \psi(x)$ satisfies the multiplicative error model
    in \eqref{eq:eigenfunction_error_bound}. Then for any compact
    sets $W, V \subseteq X$ and any $\mathcal{T} \in
    \{\mathcal{L}, \mathcal{A}\}$,
    \begin{equation*}
        -\mathcal{T}^{\varepsilon}_{V,W} \;\le\;
        \mathcal{T}^{\tilde\psi}_{W,V} - \mathcal{T}^{\psi}_{W,V}
        \;\le\; \mathcal{T}^{\varepsilon}_{W,V}.
    \end{equation*}
\end{lemma}
\vspace{-1em}
\begin{proof}
    Under \eqref{eq:eigenfunction_error_bound}, both
    $\log|\tilde\psi| = \log|\psi| + \log|\varepsilon|$ and
    $\angle\tilde\psi = \angle\psi + \angle\varepsilon$ hold. Let
    $h$ denote either $\log|\psi|$ or $\angle\psi$ depending on
    $\mathcal{T}$, with $e$ the corresponding error term, so that
    $\mathcal{T}^{\tilde\psi}_{W,V} = \sup_V (h+e) - \inf_W (h+e)$.
    Separating the supremum and infimum yields
    \begin{align*}
        \mathcal{T}^{\tilde\psi}_{W,V}
        &\le [\sup_V h + \sup_V e] - [\inf_W h + \inf_W e]
        = \mathcal{T}^{\psi}_{W,V} + \mathcal{T}^{\varepsilon}_{W,V}, \\ \vspace{-0.5em}
        \mathcal{T}^{\tilde\psi}_{W,V}
        &\ge [\sup_V h + \inf_V e] - [\inf_W h + \sup_W e]
        = \mathcal{T}^{\psi}_{W,V} - \mathcal{T}^{\varepsilon}_{V,W}.
    \end{align*}
    This completes the proof.
\end{proof}
Next, we proceed with the analysis of the impact of sampling
uncertainty on the verification framework.

\begin{lemma} \label{lemma:extremum_sampling}
    Let $h : S \to \mathbb{R}$ be Lipschitz on a non-empty compact
    $S \subseteq X$ with Lipschitz constant upper bound $\bar{L}$.
    Let $h^\star \doteq \max_{x \in
    S} h(x)$ attained at $x^\star$, and let $\tilde h^\star_N$ be
    the empirical maximum over $N$ i.i.d.\ samples from $S$.
    For any $\sigma \in (0,1)$ and any $\epsilon > 0$ small enough
    such that $\mathcal{B}_{\epsilon/\bar L}(x^\star) \subseteq S$, the inequality $\tilde h^\star_N \ge h^\star - \epsilon$ holds with probability
    at least $1 - \sigma$ if
    \begin{equation} \label{eq:N_generic_extremum}
        N \;\ge\; \ceil{ \frac{\mathrm{vol}(S)}{\omega_n}
            \left( \frac{\bar L}{\epsilon} \right)^{\! n}
        \log\frac{1}{\sigma} }.
    \end{equation}
\end{lemma}
\vspace{-1em}
\begin{proof}
    Since $h^\star - h(x) \le \bar L \|x - x^\star\|$ for all
    $x \in S$, then $\mathcal{S}_{\epsilon} \doteq \{x \in S : h^\star - h(x) \le \epsilon\}$ contains
    $\mathcal{B}_{\epsilon/\bar L}(x^\star)$, whose $\mu$-mass $P_{\epsilon}$ satisfies
    $P_{\epsilon} \geq \underline{P}_{\epsilon} \doteq \omega_n (\epsilon/\bar
    L)^n / \mathrm{vol}(S)$. The event
    $\tilde h^\star_N < h^\star - \epsilon$ requires all $N$
    samples to miss $\mathcal{S}_{\epsilon}$, occurring with
    probability at most $(1 - \underline{P}_{\epsilon})^N$.
    Enforcing $(1 - \underline{P}_{\epsilon})^N \le \sigma$ and
    using $-\log(1-p) \ge p$ for $p \in [0,1)$ yields
    \eqref{eq:N_generic_extremum}.
\end{proof}

\begin{lemma} \label{lemma:impact of sampling}
    Let $W, V \subseteq X$ be non-empty compact sets and let
    $\tilde{\mathcal{T}}^{\psi}$ denote the empirical estimate of
    $\mathcal{T}^{\psi}$ from $N$ i.i.d.\ samples drawn from $\mu$
    on each set, with $\mathcal{T} \in \{\mathcal{L}, \mathcal{A}\}$.
    Define $\bar L_{\max} \doteq \max\{\bar L_{\mathcal{L}}, \bar
    L_{\mathcal{A}}\}$
    \footnote{In our implementation, $\bar L$ is computed by scaling
    the sampling-based empirical lower bound $\underline{L}$ by an inflation factor $\rho$.}
    and $\mathrm{vol}_{\max} \doteq \max\{
    \mathrm{vol}(W), \mathrm{vol}(V)\}$. For any $\delta \in (0,1)$
    and any $\epsilon > 0$ small enough such that the ball
    $\mathcal{B}_{\epsilon/(2 \bar L_{\max})}(x^*)$ is contained in
    the set on which the corresponding extremum is attained, if
    \begin{equation} \label{eq:N_sampling_lipschitz}
        N \;\ge\; \ceil{ \frac{\mathrm{vol}_{\max}}{\omega_n} \left(
            \frac{2 \bar L_{\max}}{\epsilon} \right)^{\! n} \log
        \frac{4}{\delta} },
    \end{equation}
    then $\mathcal{T}^{\psi}_{W,V} - \epsilon \;\le\;
    \tilde{\mathcal{T}}^{\psi}_{W,V} \;\le\;
    \mathcal{T}^{\psi}_{W,V},
    \forall \mathcal{T} \in \{\mathcal{L}, \mathcal{A}\}$
    hold jointly with probability at least $1 - \delta$.
\end{lemma}
\vspace{-0.5em}
\begin{proof}
    The upper bound holds deterministically since the empirical sup
    and inf satisfy $\tilde{\overline{h}}_S \le \overline{h}_S$ and
    $\tilde{\underline{h}}_S \ge \underline{h}_S$ for any continuous
    $h$ and compact $S$. For the lower bound, let $h_{\mathcal{L}}
    \doteq \log|\psi|$ and $h_{\mathcal{A}} \doteq \angle\psi$, which are both
    Lipschitz under Assumption \ref{assumption:setup}(iii). For each
    $\mathcal{T} \in \{\mathcal{L}, \mathcal{A}\}$, define the corresponding
    extremum estimation events
    \begin{equation*}
        \mathcal{E}_{\mathcal{T}, V} \doteq \bigl\{
            \tilde{\overline{h_\mathcal{T}}}_V \ge \overline{h_\mathcal{T}}_V -
        \tfrac{\epsilon}{2} \bigr\}, \quad
        \mathcal{E}_{\mathcal{T}, W} \doteq \bigl\{
            \tilde{\underline{h_\mathcal{T}}}_W \le \underline{h_\mathcal{T}}_W +
        \tfrac{\epsilon}{2} \bigr\}.
    \end{equation*}
    Applying Lemma \ref{lemma:extremum_sampling} to $h_\mathcal{T}$ on $V$
    and to $-h_\mathcal{T}$ on $W$ at tolerance $\epsilon/2$ and risk
    $\delta/4$ bounds the failure probability of each of the four
    events by $\delta/4$.
    Replacing $\bar L_\mathcal{T},
    \mathrm{vol}(\cdot)$ by their maxima in
    \eqref{eq:N_generic_extremum} yields
    \eqref{eq:N_sampling_lipschitz}, then $N$ satisfies all four
    conditions simultaneously. On $\mathcal{E}_{\mathcal{T}, V} \cap
    \mathcal{E}_{\mathcal{T}, W}$, subtracting the two inequalities gives
    $\tilde{\mathcal{T}}^{\psi}_{W,V} \ge
    \mathcal{T}^{\psi}_{W,V} - \epsilon$ for $\mathcal{T}$. A union
    bound over the four failure complements gives joint probability
    at least $1 - \delta$, on which
    the lemma's claim
    holds
    simultaneously for both $\mathcal{T} \in \{\mathcal{L}, \mathcal{A}\}$.
\end{proof}
\begin{theorem} \label{theorem:hausdorff}
    Suppose $V, W \subseteq X$ are non-empty compact sets and let
    $\tilde{I}_{V,W}$ denote the empirical reach-time interval from
    $V$ to $W$ estimated from principal eigenpairs
    $\{(\tilde\lambda_i, \tilde\psi_i)\}_{i=1}^m$ under Assumption
    \ref{assumption:setup}, with $\tilde\psi_i = \varepsilon_i
    \psi_i$. Define
    \begin{equation*}
        \Lambda^{\mathcal{L}}_i \doteq \Real(\lambda_i), \
        \Lambda^{\mathcal{A}}_i \doteq \Imag(\lambda_i), \
        \Delta_{\tau}^{\varepsilon_i} \doteq \max\bigl\{
            |\mathcal{T}^{\varepsilon_i}_{V,W}|,\,
        |\mathcal{T}^{\varepsilon_i}_{W,V}| \bigr\},
    \end{equation*}
    for $\tau \in \{\mathcal{L}, \mathcal{A}\}$, and
    $\Delta_{\tau}^{\varepsilon} \doteq \max_{i=1}^m
    \Delta_{\tau}^{\varepsilon_i}$.
    For any tolerance $\epsilon > 0$ satisfying the smallness
    condition of Lemma \ref{lemma:impact of sampling} and risk
    $\delta \in (0,1)$, if
    \begin{equation} \label{eq:N_theorem_hausdorff}
        N \;\ge\; \ceil{ \frac{\mathrm{vol}_{V,W}}{\omega_n} \left(
            \frac{2 \bar L_{\max}}{\epsilon} \right)^{\! n} \log
        \frac{4m}{\delta} }
    \end{equation}
    with $\mathrm{vol}_{V,W} \doteq \max\{\mathrm{vol}(V),
    \mathrm{vol}(W)\}$, then
    \begin{equation*}
        \mathbb{P}\bigl\{ d_H(\hat{I}_{V,W}, \tilde{I}_{V,W}) \le
        \Delta \bigr\} \;\ge\; 1 - \delta,
    \end{equation*}
    with $\Delta \doteq \max_{\tau \in \{\mathcal{L},\mathcal{A}\}}
    (\epsilon + \Delta_{\tau}^{\varepsilon})/\min_i
    |\Lambda^{\tau}_i|$.
\end{theorem}
\begin{proof}
    For each $i \in \{1, \dots, m\}$, let $\mathcal{E}_i$ denote
    the event under which both Lemma
    \ref{lemma:impact of inaccurate eigenpairs} and Lemma
    \ref{lemma:impact of sampling} hold for $(\tilde\lambda_i,
    \tilde\psi_i)$ on $(V,W)$, and let $\mathcal{E} :=
    \bigcap_{i=1}^m \mathcal{E}_i$.
    On event $\mathcal{E}$, the Hausdorff distance for the reach-time interval estimation is essentially
    the maximum endpoint deviations of the reach-time intervals under the inaccurate eigenpairs.
    For each possible ordering $(P,Q) \in \{(V,W),(W,V)\}$, the endpoint deviation can be bounded as
    \begin{align} \label{eq:endpoint deviation bound}
        \biggl| \frac{\sum_i \alpha_i
            (\tilde{\mathcal{T}}^{\tilde\psi_i}_{P,Q} -
        \mathcal{T}^{\psi_i}_{P,Q})}{\sum_i \alpha_i
        \Lambda^{\mathcal{T}}_i} \biggr|
        &= \biggl| \frac{\sum_i \alpha_i \Lambda^{\mathcal{T}}_i
            (\tilde{\mathcal{T}}^{\tilde\psi_i}_{P,Q} -
        \mathcal{T}^{\psi_i}_{P,Q}) / \Lambda^{\mathcal{T}}_i}{\sum_i \alpha_i
        \Lambda^{\mathcal{T}}_i} \biggr| \notag \\
        \leq \max_i &\frac{|\tilde{\mathcal{T}}^{\tilde\psi_i}_{P,Q} -
        \mathcal{T}^{\tilde\psi_i}_{P,Q}| + |\mathcal{T}^{\tilde\psi_i}_{P,Q} - \mathcal{T}^{\psi_i}_{P,Q}|}{| \Lambda^{\mathcal{T}}_i |}.
    \end{align}
    On $\mathcal{E}_i$, applying Lemma
    \ref{lemma:impact of sampling} to $\tilde\psi_i$ at tolerance
    $\epsilon$ and risk $\delta/m$, together with Lemma
    \ref{lemma:impact of inaccurate eigenpairs} for $\tilde\psi_i =
    \varepsilon_i \psi_i$ yields
    \begin{equation*}
        \bigl| \tilde{\mathcal{T}}^{\tilde\psi_i}_{P,Q} -
        \mathcal{T}^{\tilde\psi_i}_{P,Q} \bigr| \le \epsilon, \qquad
        \bigl| \mathcal{T}^{\tilde\psi_i}_{P,Q} -
        \mathcal{T}^{\psi_i}_{P,Q} \bigr| \le
        \Delta_{\mathcal{T}}^{\varepsilon_i}.
    \end{equation*}
    Substituting both bounds into \eqref{eq:endpoint deviation bound} gives,
    \begin{equation*}
        \biggl| \frac{\sum_i \alpha_i
        \tilde{\mathcal{T}}^{\tilde\psi_i}_{P,Q}}{\sum_i \alpha_i
        \Lambda^{\mathcal{T}}_i} - \frac{\sum_i \alpha_i
        \mathcal{T}^{\psi_i}_{P,Q}}{\sum_i \alpha_i
        \Lambda^{\mathcal{T}}_i} \biggr|
        \;\le\; \frac{\epsilon +
        \Delta_{\mathcal{T}}^{\varepsilon}}{\min_i
        |\Lambda^{\mathcal{T}}_i|} = \Delta.
    \end{equation*}
    Taking the maximum over $\mathcal{T} \in \{\mathcal{L}, \mathcal{A}\}$
    and using $\tilde{I}_{V,W} \subseteq \tilde{I}^{mag}_{V,W} \cap
    \tilde{I}^{phase}_{V,W}$ gives $d_H(\hat{I}_{V,W},
    \tilde{I}_{V,W}) \le \Delta$ on $\mathcal{E}$.
    Since event $\mathcal{E}$ is only a sufficient condition to guarantee the
    bounded Hausdorff distance, thus
    \begin{align}
        &\mathbb{P}\left\{ d_{H} (\tilde{I}(\lambda,\psi),\hat{I}(\lambda,\psi))
        \leq \Delta \right\}
        \geq \mathbb{P}\left( \bigcap_{i=1}^{m} \mathcal{E}_{i} \right) \notag \\
        &= 1- \mathbb{P}\left( \bigcup_{i=1}^{m} \mathcal{E}_{i}^{c} \right)
        \geq 1 - \sum_{i=1}^{m} \mathbb{P}(\mathcal{E}_{i}^{c}). \label{eq:probl_ineq}
    \end{align}
    Applying Lemma \ref{lemma:impact of sampling} to each
    $\tilde\psi_i$ at tolerance $\epsilon$ and risk $\delta/m$, the
    minimal samples amount in \eqref{eq:N_sampling_lipschitz} reduces to
    \eqref{eq:N_theorem_hausdorff} and yields
    $\mathbb{P}(\mathcal{E}_i^c) \le \delta/m$, so
    \eqref{eq:probl_ineq} gives $\mathbb{P}\{d_H \le \Delta\} \ge 1
    - \delta$.
\end{proof}
\color{black}

\section{Experiments}
In this section, we present numerical experiments to validate the proposed framework,
analyze its convergence properties, and highlight its key advantages over conventional
set-based methods.
The principal Koopman eigenpairs are learned from data
with ResDMD using polynomial basis functions in our implementation.


\begin{example}
    (System with known eigenfunctions)
    Consider a 2D nonlinear system defined by
    $$
    \begin{bmatrix}\dot{x}_{1} \\ \dot{x}_{2}
    \end{bmatrix} =
    \left[\nabla \Psi(x)\right]^{-1}
    \begin{bmatrix}-1&0 \\ 0&2.5
    \end{bmatrix} \Psi(x), $$
    where the analytical principal eigenfunctions at the unstable equilibrium $(0,0)$ are
    $\Psi(x) = [\psi_{1}(x), \psi_{2}(x)]^\top$, with $\psi_{1}(x) = x_{1}^{2}+2x_{2}+x_{2}^{3}$
    and $\psi_{2}(x)=x_{1}+\sin(x_2)+x_{1}^{3}$ for $\lambda_{1}=-1, \lambda_{2}=2.5$.
    Using 1000 trajectories uniformly initialized in $[-0.5, 2.5] \times [-1.5, 1.5]$ with simulations
    step $\Delta_t = 0.05$ for 10 steps,
    we apply the ResDMD algorithm to learn the
    approximated
    eigenfunctions
    $\tilde{\psi}_1(x)$ and $\tilde{\psi}_2(x)$.
    The task here is to verify reachability from an initial set $X_0=\{ x \mid h_{0.05, 1.15, 1, 2, 0.05}(x) \leq -0.1 \}$
    to a target set $X_F=\{ x \mid h_{1.85, -0.75, 5, 8, 0.1}(x) \leq -0.7 \}$, where
    $$ h_{c_1,c_2,a,b,s}(x) = -\left[ 1 - \frac{x_{1}-c_1}{3} + a z_2^5 + b z_1^3 \right] e^{-(z_1^2 + z_2^2)}$$ with
    $z_i = (x_i - c_i)/s$.
    Our method yields an
    empirical
    reach-time bound $\tilde{I} = [0.70, 0.97]$,
    which successfully verifies the reachability as confirmed in Fig. \ref{fig:rv_combined}.
    To validate the probabilistic guarantees of Theorem \ref{theorem:hausdorff},
    we evaluates 1000 independent trials across varying tolerances $\epsilon \in
    \{0.05, 0.02\}$ and sample sizes $N \in \{N^\star, 2N^\star\}$, where
    $N^\star$ is the minimum size required for risk $\delta=0.01$.
    As shown in Fig. \ref{fig:ex22}, all Hausdorff distances remain strictly below the theoretical bounds,
    thus robustly satisfying the prescribed 99\% confidence guarantee with $\delta = 0.01$.
    Additionally, decreasing $\epsilon$ tightens both the theoretical bound
    and the empirical error spread, while doubling $N$ further concentrates
    the errors toward zero, demonstrating improve both parameters can
    enhance estimation precision.

    \begin{figure}[htbp!]
        \centering
        \begin{subfigure}[t]{0.49\linewidth}
            \centering
            \includegraphics[width=\linewidth]{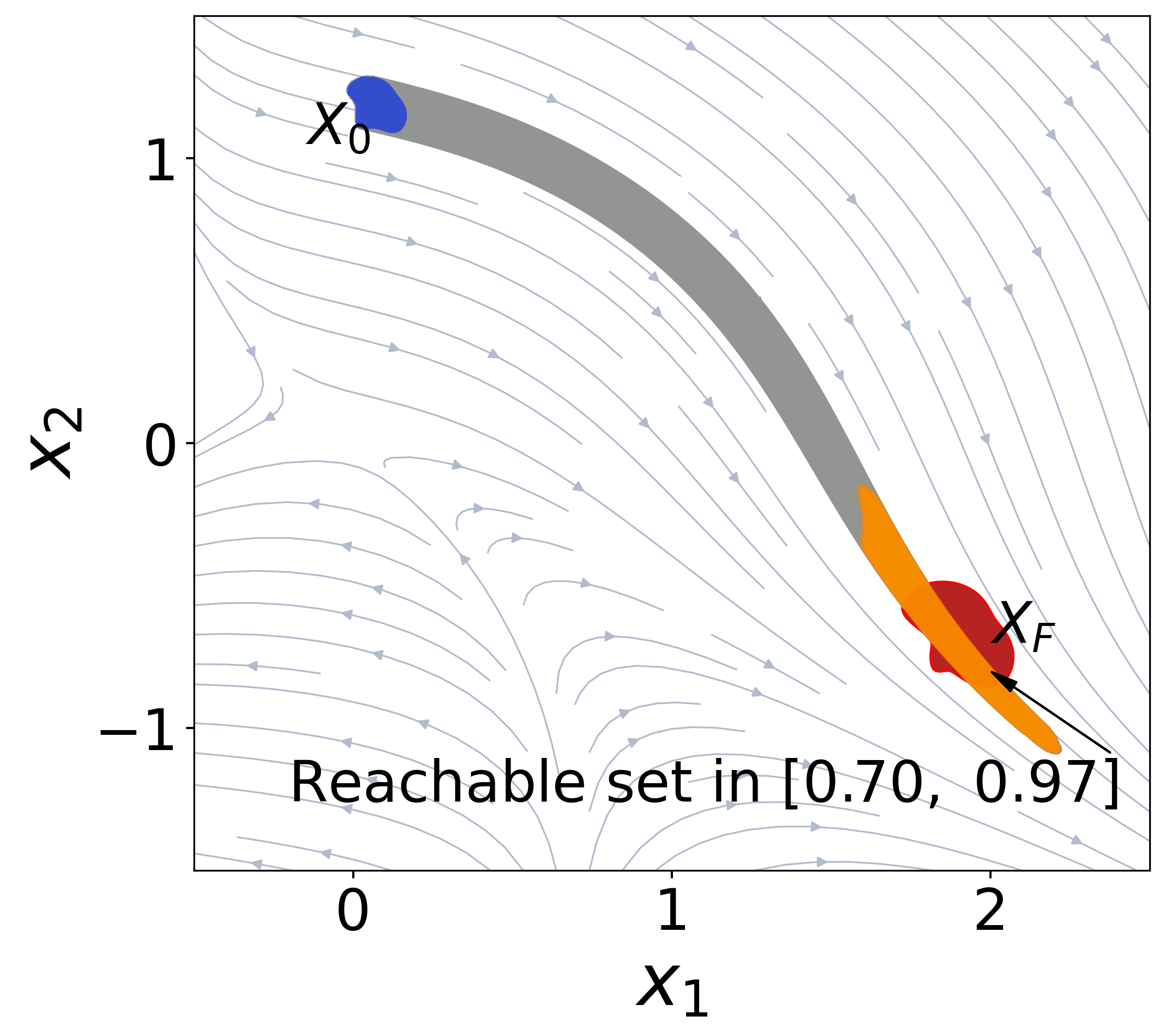}
            \caption{Example 1, $\tilde{I}=[0.70,0.97]$.}
            \label{fig:rv_nl_eig}
        \end{subfigure}\hfill
        \begin{subfigure}[t]{0.49\linewidth}
            \centering
            \includegraphics[width=\linewidth]{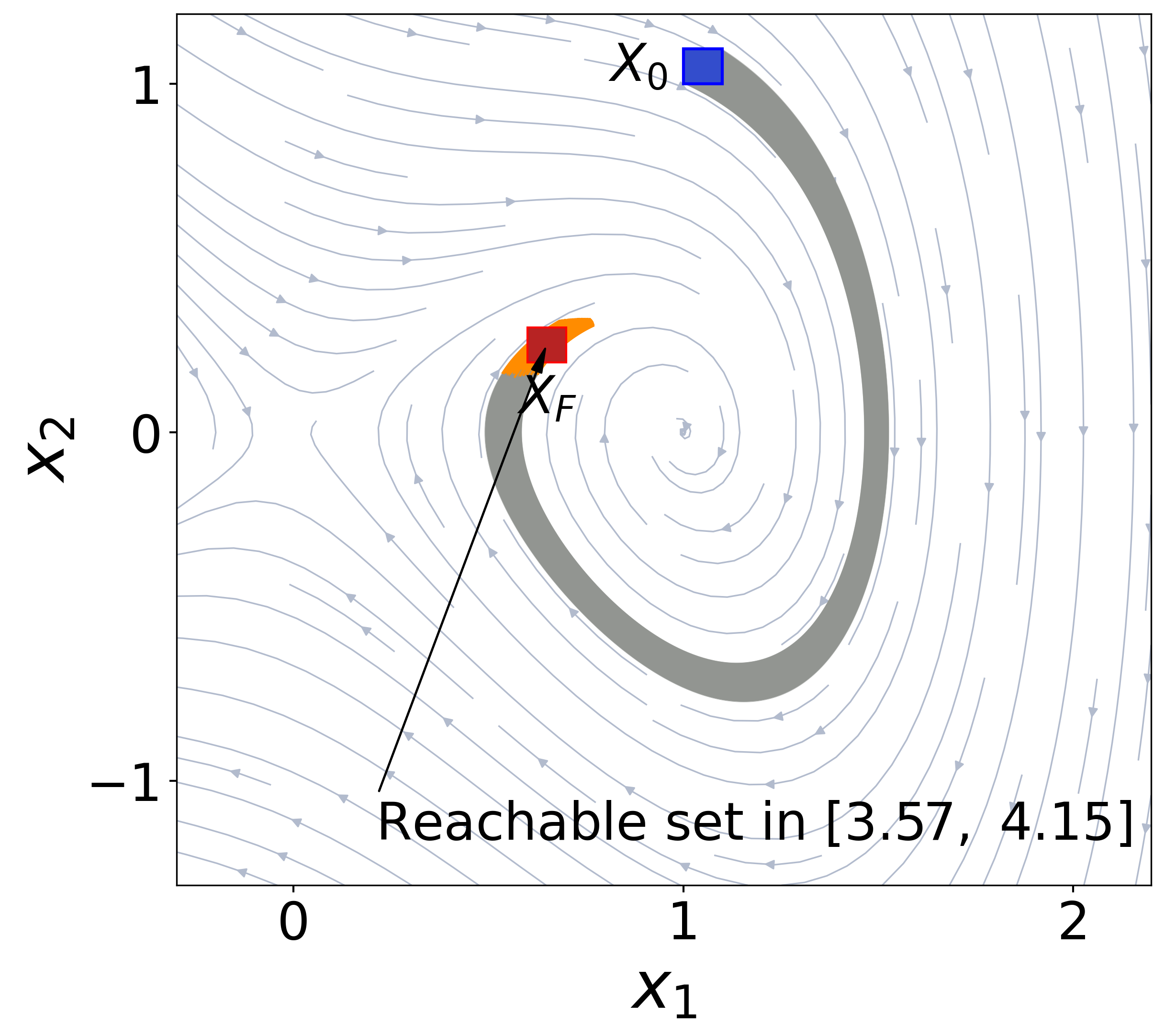}
            \caption{Example 2, $\tilde{I}=[3.57,4.15]$.}
            \label{fig:duff_rv}
        \end{subfigure}
        \caption{Reachable sets from $X_0$ for Examples 1 and 2 over
            the estimated reach-time bounds (orange), alongside the full
        reachable sets up to the lower endpoint (gray).}
        \label{fig:rv_combined}
    \end{figure}


    \begin{figure}[htbp!]
        \centering
        \includegraphics[width=\linewidth]{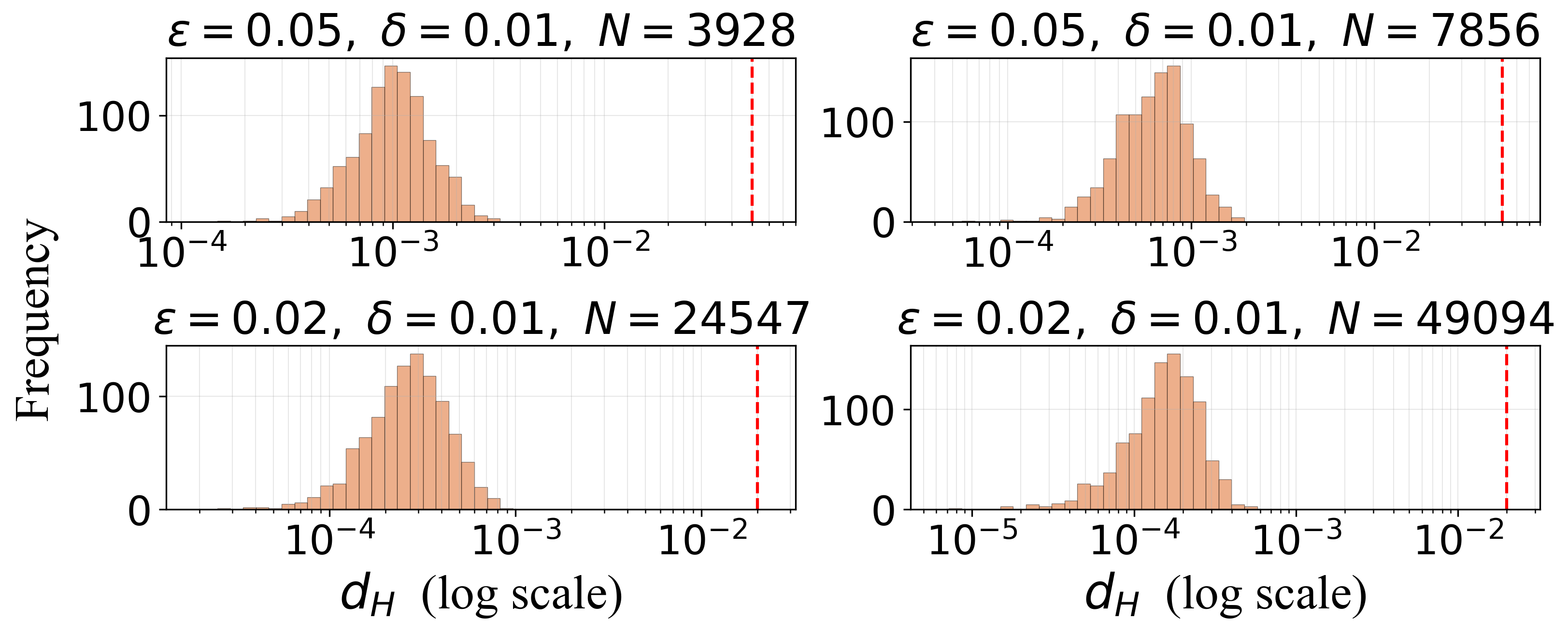}
        \caption{Hausdorff distances over 1000 trials per configuration. Dashed red lines denote theoretical bounds.}
        \label{fig:ex22}
    \end{figure}
\end{example}

Next, we investigate the convergence of our framework with respect to
the quality of the data-driven inputs.


\begin{example}
    (Duffing's oscillator) Given the nonlinear dynamics
    \[
        \left[
            \begin{array}{c}
                \dot{x}_1 \\
                \dot{x}_2
        \end{array}\right] = \left[
            \begin{array}{c}
                {x}_2                  \\
                -0.5x_2 - x_1(x_1^2-1)
        \end{array}\right],
    \]
    with stable equilibrium points at $(\pm 1,0)$ and a saddle
    equilibrium point at $(0,0)$.
    For reachability verification, we consider an initial set $X_0 =
    [1.0,1.1]\times[1.0,1.1]$ and a target set $X_F = [0.6,0.7]\times[0.2,0.3]$.
    A high-fidelity baseline reach-time bound $I=[3.57, 4.15]$ is
    established by
    evaluating
    the true principal
    eigenfunctions via the path-integral method
    \citep{deka2023path}, as validated in Fig. \ref{fig:duff_rv}.
    To validate Theorem~\ref{theorem:hausdorff}
    under varying qualities of eigenfunction approximations, we fix
    $(\epsilon,\delta)=(0.02,0.05)$,
    and modulate the approximation quality
    across basis degrees $\{8,10,12,14\}$ and trajectory
    numbers $\{200,500,1000,2000\}$.
    bound and the empirical Hausdorff distance. As illustrated in Fig.~\ref{fig:hd_cov},
    for all $1000$ independent trials per configuration, the margin $\Delta-d_H$
    remains non-negative exceeding the prescribed $95\%$ confidence level,
    demonstrating that the probabilistic bounds
    hold robustly regardless of the underlying data-driven approximation quality.
    Notably, eigenfunction approximations using
    polynomials of degree 8 outperform other bases for certain
    settings, as ResDMD captures average trajectory consistency rather than
    local pointwise accuracy.
    This empirically validates a key insight of Theorem \ref{theorem:hausdorff},
    that is, the estimated reach-time bounds is highly sensitive
    to the worst-case pointwise error that captured by the $L_\infty$-norm.

    \begin{figure}[htbp!]
        \centering
        \includegraphics[width=\linewidth]{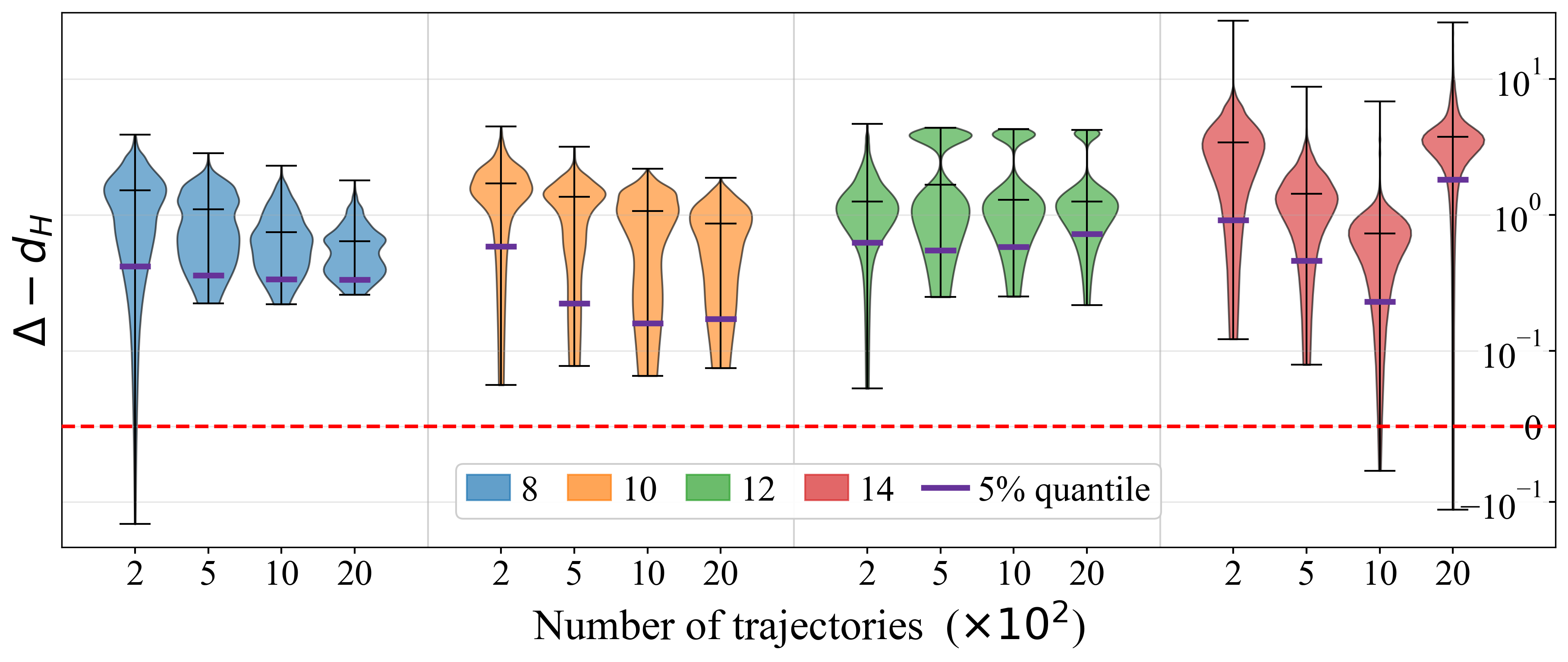}
        \caption{
            Distributions of $\Delta - d_H$ across varying
            settings.
        }
        \label{fig:hd_cov}
    \end{figure}


\end{example}
Finally,
we compare our method
against a
Koopman-based reachability analysis approach
developed in \citep{bak2025reachability} that relies on
set-propagation technique.

\begin{example}
    (Roessler attractor) Consider the dynamics
    \begin{align*}
        \begin{bmatrix}
            \dot{x}_1 \\
            \dot{x}_2 \\
            \dot{x}_3
        \end{bmatrix} =
        \begin{bmatrix}
            -x_2 - x_3 \\
            x_1 + 0.2 x_2 \\
            0.2 + x_3(x_1-5.7)
        \end{bmatrix}.
    \end{align*}
    We formulate the task to verify no trajectories starting from the
    initial set $X_0 = [-0.5, 0.5] \times [-9.0, -8.0] \times [-0.5,
0.5]$ can reach the target set $X_F= 10.5, 11] \times [-4.4, -3.9]
\times [-0.6, -0.1]$ within a given time horizon $[0,1]$.
The competitor method lifts the 3-dimensional state into a
73-dimensional observable space, computes the reachable set of the
resulting linear system using Polynomial Zonotope,
the computed set is then projected back to the original state space for verification.
As depicted in Fig. \ref{fig:roessler_rv}, the set-propagation method
yields
a conservative over-approximated
that intersects with the
target set,
rendering the unreachability verification inconclusive.
In contrast, our framework leverages
the principal Koopman eigenpairs learned from trajectory data to
bypass the conservative and expensive geometric operations, and our
estimated reach-time bound is empty with high
probability. This conclusion is validated by the simulated
trajectories as shown in Fig. \ref{fig:roessler_rv}.

\begin{figure*}[htbp!]
    \centering
    \includegraphics[width=0.95\linewidth]{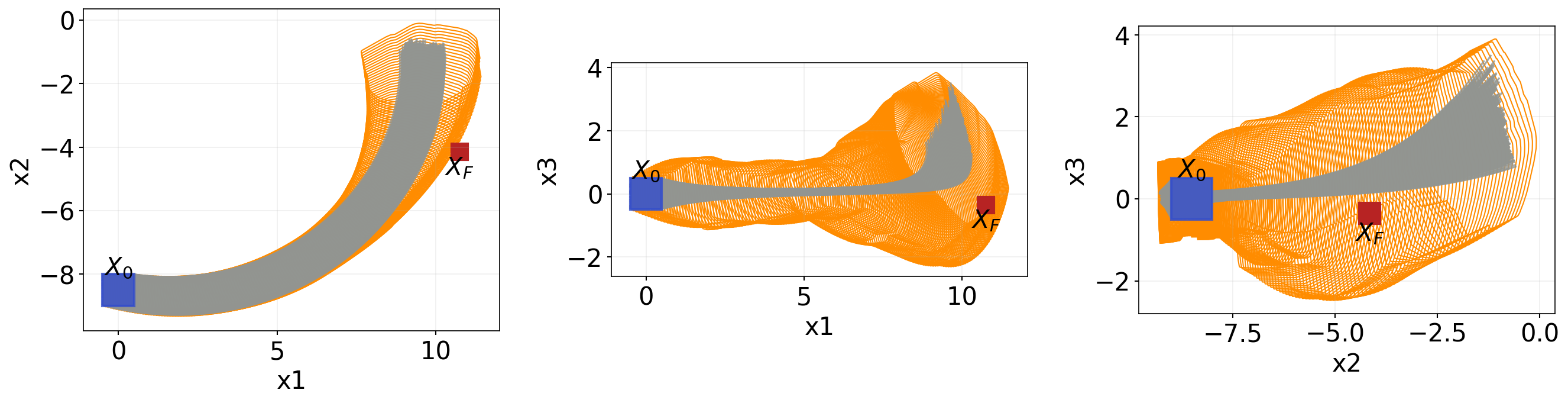}
    \caption{Reachability verification of the  Roessler attractor system.}
    \label{fig:roessler_rv}
\end{figure*}
\end{example}


\section{Conclusions}

This paper presents
a novel data-driven framework for
reachability verification for unknown dynamical systems
using the
Koopman spectrum.
Our main contribution is propagating model uncertainty from data into
a formal probabilistic guarantee on the time-to-reach bounds.
Future work will incorporate Koopman eigenvalue
uncertainty into the analysis and extend this framework to control synthesis
for safety-critical systems.


\section{DECLARATION OF GENERATIVE AI AND AI-ASSISTED TECHNOLOGIES IN THE WRITING PROCESS}

During the preparation of this work, the authors used Gemini for linguistic
polishing, and Claude Code to assist numerical experiments. After using these
tools, the authors reviewed and edited the content as needed and take full
responsibility for the content of the publication.

\bibliography{refs}             




\end{document}